\documentclass[11pt,a4paper]{article}
\usepackage{indentfirst,mathrsfs}
\usepackage{amsfonts,amsmath,amssymb,amsthm}
\usepackage{latexsym,amscd}
\usepackage{amsbsy}
\usepackage{makecell}
\usepackage{enumerate}
\usepackage[numbers,sort&compress]{natbib}
\usepackage{color}
\newtheorem{theorem}{Theorem}

\newtheorem{lem}{Lemma}
\newtheorem{cor}{Corollary}

\newtheorem{example}{Example}
\newtheorem{remark}{Remark}

\begin{document}
\begin{center}{\bf \Large
Four classes of optimal $p$-ary cyclic codes}
\end{center}
\begin{center}

{\small Jinmei Fan, Jingyao Feng, Yuhan Men, Yanhai Zhang \footnote{ Corresponding author.} \footnote{This paper has been submitted to Designs, Codes and Cryptography on 21 January 2025, Submission ID 6d905dd2-eb91-4c41-a0d8-ca66472e7960.}\footnote{\small J. Fan's research was supported by the National Natural Science Foundation of China under Grant 12061027.
Y. Zhang's  research was supported by  the
Doctoral Research Foundation of Guilin University of Technology under Grant GUTQDJJ2018033 and by the Opening Fund of Key Laboratory of Cognitive Radio and Information Processing, Ministry of Education  under Grant CRKL210206.}
\footnote{\textit{E-mail addresses}: 2007027@glut.edu.cn (J. Fan), 1244574837@qq.com (J. Feng), 3231929859@qq.com (Y. Men), zhang.yanhai@foxmail.com
 (Y. Zhang).
}\\
\medskip School of Mathematics and Statistics, Guilin University of Technology, Guilin 541004, China \\}

\end{center}
\begin{quote}
{\small {\bf Abstract:}} Let $p\geq5$ be an odd prime and $m$ be a positive integer. Little progress on the study of optimal $p$-ary cyclic codes with parameters $[p^m-1,p^m-2m-2,4]$ has been made. In this paper, by  weakening the necessary and sufficient conditions on  cyclic codes to have codewords of Hamming weight 3 and analyzing the solutions of certain equations over finite fields, four classes of optimal $p$-ary cyclic codes deduced by $\frac{p^m+1}{2}$ with parameters $[p^m-1,p^m-2m-2,4]$ are presented.
Wherein three  classes of optimal $p$-ary cyclic codes are infinite.
Many classes of known optimal quinary cyclic codes  with parameters $[5^m-1,5^m-2m-2,4]$ are special cases of the codes constructed in this paper.


\end{quote}

\section{Introduction}

Denote by $\mathbb{F}_{p^{m}}$ the finite field with $p^{m}$ elements  and $\mathbb{F}^*_{p^{m}}=\mathbb{F}_{p^{m}}\setminus \{0\}$, where $p$ is a prime and $m$ is a positive integer.
An $[n,k,d]$ linear code $\textit{C}$ over $\mathbb{F}_p$ is a $k$-dimension  subspace of $\mathbb{F}_{p}^n$ with  minimum (Hamming) nonzero weight $d$. It is
called \textit{cyclic}  if $(c_0,c_1,\cdots,c_{n-1})\in \textit{C}$ implies $(c_1,c_2,\cdots,c_{n-1},c_0)\in \textit{C}$. Let ${\rm gcd}(n,p)=1$. Any  code $\textit{C}$ of length $n$ over $\mathbb{F}_{p}$ corresponds to a subset of   $\mathbb{F}_{p}[x]/(x^n-1)$ by identifying any  vector
$(c_0,c_1,\cdots,c_{n-1})\in \mathbb{F}_{p}^n$ with $$c_0+c_1x+c_2x^2+\cdots +c_{n-1}x^{n-1}\in \mathbb{F}_{p}[x]/(x^n-1).$$ The linear code $\textit{C}$ is cyclic if and only if the corresponding subset of $\mathbb{F}_{p}[x]/(x^n-1)$ is an ideal.
Since every ideal of
$\mathbb{F}_{p}[x]/(x^n-1)$ is principal, any cyclic code $\textit{C}$ can be expressed as  $\textit{C}=( g(x))$, where $g(x)$ is the monic polynomial of least degree in the code. Then $g(x)$ is called the \textit{generator polynomial}  and  $(x^n-1)/g(x)$ is referred to as the \textit{parity-check polynomial} of $\textit{C}$.
Let $\alpha$ be a generator of $\mathbb{F}^*_{p^{m}}$ and denote the minimal polynomial of $\alpha^i$ over $\mathbb{F}_{p}$ by $m_{\alpha^i}(x)$. The $p$-ary cyclic codes with generator polynomial $m_{\alpha^{i_1}}(x)m_{\alpha^{i_2}}(x)\cdots m_{\alpha^{i_r}}(x)$  are denoted by $C_p(i_1,i_2,\cdots,i_r)$.

Cyclic codes are an important subclass of linear codes and the most studied of all codes as they include the important family of BCH codes and are easy to encode. Furthermore, cyclic codes are building blocks for many other codes, such as Kerdock, Preparata and Justesen codes \cite{MacWilliams}.
Cyclic codes are widely used  in consumer electronics, data storage systems and communication systems due to their desirable algebraic properties and efficient algorithms for encoding and decoding processes \cite{Chien, Forney, Prange}. Much progress had been made on cyclic codes in the past few decades. For some recent advances, we refer the reader to \cite{CDY, DH, LLHD, LLHDT, open1, FLZ, WW, YZD, open2, Zha1, open3, Zha2, LCL, ZLS, XCX, Fan, Liu1, WLZ}.
In recent years, optimal cyclic codes over finite fields with respect to the Sphere Packing bound have attracted researchers' widespread attention. By the Sphere Packing bound, ternary cyclic codes
with parameters $[3^m-1,3^m-2m-1,4]$ was proved to be optimal  \cite{DH}.
A family of optimal ternary cyclic codes $C_3(1,v)$ with parameters $[3^m-1,3^m-2m-1,4]$ has been extensively studied \cite{DH, CDY,LLHD, LLHDT, open1, FLZ, WW, YZD, open2, Zha1, open3, Zha2, LCL, ZLS}, where $x^v$ is the monomial over $\mathbb F_3$ including perfect nonlinear (PN) monomials and almost perfect nonlinear (APN) monomials.
Furthermore, two classes of optimal subcodes of $C_3(1,v)$ with parameters $[3^m-1,3^m-2m-2,5]$ have been also studied. One class of  optimal subcodes of $C_3(1,v)$ with parameters $[3^m-1,3^m-2m-2,5]$, denoted by $C_3(0,1,v)$,  is closely related to the codes investigated in \cite{CDY} and \cite{YCD}. The other class of  optimal subcodes $C_3(\frac{3^m-1}{2},1,v)$ of $C_3(1,v)$ with parameters $[3^m-1,3^m-2m-2,5]$ was investigated by analyzing irreducible factors of certain polynomial with low degrees in \cite{LLHDT}.

Let $p\geq 5$ be an odd prime. When $|C_v|=m$, cyclic codes $C_p(1,v)$ have minimum distance 2 or 3 which  may not be interesting \cite{DH}. According  to the Sphere Packing bound, $p$-ary cyclic codes with parameters $[p^m-1,p^m-2m-2,4]$ are optimal. In 2016, to obtain optimal $p$-ary cyclic codes related to $C_p(1,v)$, researchers began to investigate
subcodes  of   $C_p(1,v)$ with parameters $[p^m-1,p^m-2m-2,4]$ \cite{XCX}.  Several classes of optimal $p$-ary cyclic codes $C_p(0,1,v)$  with  parameters $[p^m-1,p^m-2m-2,4]$ were obtained from PN monomials and the inverse functions \cite{XCX}. In 2020, to obtain more optimal $p$-ary cyclic codes with parameters $[p^m-1,p^m-2m-2,4]$, another class of subcodes $C_p(\frac{p^m-1}{2},1,v)$ of $C_p(1,v)$ was presented \cite{Fan}.
A necessary and sufficient condition on $v$ for the quinary cyclic code $C_5(\frac{5^m-1}{2},1,v)$  with parameters $[5^m-1,5^m-2m-2,4]$  to be optimal was given in \cite{Fan}. Based on the necessary and sufficient condition, several classes of optimal quinary cyclic codes $C_5(\frac{5^m-1}{2},1,v)$ were presented in \cite{Fan}. Furthermore, two classes of optimal quinary cyclic codes $C_5(0,1,v)$ and  $C_5(\frac{5^m-1}{2},1,v)$ were investigated by discussing the solutions of certain equations  in 2020 \cite{Liu1}.
In 2023, by analyzing the solutions of certain equations over finite fields, four classes of optimal $p$-ary cyclic codes $C_p(\frac{p^m-1}{2},1,v)$ and two classes of optimal quinary cyclic codes $C_5(\frac{5^m-1}{2},1,v)$ were presented  in \cite{WLZ}.  We summarize known optimal $p$-ary cyclic codes $C_p(0,1,v)$ with parameters $[p^m-1,p^m-2m-2,4]$ in Table 1.

\begin{table}
\centering
\caption{ Known  optimal cyclic codes $C_p(0,1,v)$  with parameters $[p^m-1, p^m-2m-2,4]$}
\label{Tab:01}
\begin{tabular}{ccc}
  \hline
  \makecell[l]{Vales of $v$ or requirements on $v$}  &  \makecell[l]{Conditions} & Reference \\\hline
   \makecell[l]{$v=p^k+1$} & \makecell[l]{$k\neq \frac{m}{2}$, $p\geq 5$ is an odd prime} & [19] \\
  \makecell[l]{$v=p^m-2$} &  \makecell[l]{$m>1$, $p\geq 5$ is an odd prime} & [19] \\
 \makecell[l]{$v$ is PN or APN over $\mathbb {F}_{5^m}$} &  \makecell[l]{$m>1$} & [19] \\
 \makecell[l]{$v=\frac{5^m-1}{2}+3$} &  \makecell[l]{$m$ is even, ${\rm gcd}(k,2m)=1$}&[19]  \\
 \makecell[l]{$v=\frac{5^m-1}{2}+\frac{5^k+1}{2}$} & \makecell[l]{$m$ is even, ${\rm gcd}(k,2m)=1$}&[19]  \\
 \makecell[l]{$3v\equiv 2\cdot5^k\,({\rm mod}\,5^m-1)$}& \makecell[l]{$m\geq 3$ is odd}&[19] \\
\makecell[l]{$(5^m-2)v\equiv 2\cdot5^k\,({\rm mod}\,5^m-1)$}&\makecell[l]{$m\geq 3$ is odd}&[19]  \\
\makecell[l]{$v=\frac{5^m-1}{2}-1$} & \makecell[l]{$m\equiv 0\,({\rm mod}\,2)$}&[19] \\
\makecell[l]{$v=\frac{5^m-1}{2}-3$ or  $v=14$} & \makecell[l]{$m\equiv 1\,({\rm mod}\,2)$}&[19] \\
\makecell[l]{$v(5^h+1)\equiv 5^k+1\,({\rm mod}\,5^m-1)$}& \makecell[l]{${\gcd}(m,h+k)={\gcd}(m,k-h)=1$,\\ $v\equiv3\,({\rm mod}\,4)$}&[21] \\
\makecell[l]{$v(5^h-1)\equiv 5^k-1\,({\rm mod}\,5^m-1)$}&\makecell[l]{${\gcd}(m,h)={\gcd}(m,k)=1$,\\${\gcd}(m,k-h)=1$,
$v\equiv2,3\,({\rm mod}\,4)$}&[21] \\
\hline
\end{tabular}
\end{table}
Let  $u=\frac{p^m+1}{2}$. Little progress on the study of the optimal  $p$-ary cyclic code $C_p(0,1,v)$ with parameters $[p^m-1,p^m-2m-2,4]$ has been made. In order to enrich the research results of optimal $p$-ary cyclic codes with parameters $[p^m-1,p^m-2m-2,4]$, this paper is devoted to obtaining  optimal $p$-ary cyclic codes with parameters $[p^m-1,p^m-2m-2,4]$. By weakening conditions on cyclic codes and analyzing the solutions of certain equations over finite fields, four classes of optimal $p$-ary cyclic codes $C_p(0,1,u^{-1}v)$ or $C_p(0,1,uv^{-1})$ with parameters $[p^m-1,p^m-2m-2,4]$ are presented in this paper.
Wherein three of the four classes of optimal $p$-ary cyclic codes are infinite.
The rest of this paper is organized as follows. Some preliminaries are introduced in
Section 2. Optimal $p$-ary cyclic  codes $C_p(0,1,u^{-1}v)$ or $C_p(0,1,uv^{-1})$  are presented in Section 3. Section 4 concludes this paper.

\section{Preliminaries}

In this section, we  introduce $p$-cyclotomic cosets and some basic lemmas that will be employed in subsequent sections.

\noindent \textit{A. The $p$-cyclotomic cosets modulo $p^m-1$}

The  \textit{$p$-cyclotomic coset} modulo $p^m-1$ containing $i$ is defined by
$$C_i=\{i,ip,\cdots,ip^{l_i-1}\}\subset \mathbb{Z}_{p^m-1},$$ where $l_i$ is the smallest positive integer
such that $ip^{l_i}\equiv i\,({\rm mod}\,p^m-1)$  and is called the \textit{length} of $C_i$.
Namely $l_i=|C_i|$, where $|C_i|$ is the size of the set $C_i$. The smallest positive integer in $C_i$ is called
the \textit{coset leader} of $C_i$. By definition, we have $$\bigcup\limits_{i\in \Gamma} C_i=\mathbb{Z}_{p^m-1},$$ where $\Gamma$ is the set of all the coset leaders.

The following lemmas are useful in the sequel.

\begin{lem}\label{lem1} $($\cite[Lemmas  1  and   2]  {XCX}$)$ Let $p$ be a prime and $m$ be a positive integer. For any $1\leq v\leq p^m-2$,  then $|C_v|=m$ if ${\rm gcd}(v,p^m-1)< p$ or ${\rm gcd}(v,p^m-1){\rm gcd}(p^j-1,p^m-1)\not\equiv0\,({\rm mod}\,p^m-1)$ for all $1\leq j<m$.
\end{lem}

\begin{lem}\label{gong} $($\cite[Proposition  6.2]{Gong}$)$ Let  $1\leq v\leq p^m-2$. Then
 $|C_v| \big| m$.
\end{lem}

\begin{lem}\label{gongyinzi}(\cite[Lemma 2]{FXXZ}) Let $a,\,t$ and $l$ be positive integers. Then
$${\rm gcd}(a^{t}+1,a^{l}-1)=\left\{\begin{array}{ll}
2,&{ \rm if}\,\,\frac{l}{{\rm gcd}(t,l)} \,\,{\rm is\,\,odd\,\,and}\,\,a\,\,{\rm is \,\,odd},\\
a^{{\rm gcd}(t,l)}+1,&{\rm if}\,\,\frac{l}{{\rm gcd}(t,l)} \,\,{\rm is\,\,even};\\
\end{array}\right.$$ and
$${\rm gcd}(a^{t}+1,a^{l}+1)=\left\{\begin{array}{ll}
2,&{\rm if\,\,one\,\,of}\,\,\frac{t}{{\rm gcd}(t,l)},\,\,\frac{l}{{\rm gcd}(t,l)} \,\,{\rm is\,\,even\,\,and}\,\,a\,\,{\rm is \,\,odd},\\
a^{{\rm gcd}(t,l)}+1,&{\rm if\,\,both}\,\,\frac{t}{{\rm gcd}(t,l)}\,\,{\rm and}\,\,\frac{l}{{\rm gcd}(t,l)}\,\,{\rm are\,\,odd}.
\end{array}\right.$$
\end{lem}

\begin{lem}\label{irreducible}$($\cite[Corollary 3.47]{Lidl}$)$
An irreducible polynomial over $\mathbb{F}_{p^m}$ of degree $n$ remains irreducible over $\mathbb{F}_{p^{ml}}$ if and only if gcd$(l,n)=1$.\end{lem}


\noindent \textit{B. Some useful properties about $p$-ary cyclic codes $C_p(0,\frac{p^m+1}{2},v)$}

\begin{lem} \label{lem2}
Let $p$ be an odd prime and  $u=\frac{p^m+1}{2}$. Then

\noindent 1) $|C_u|=m$; and

\noindent 2) $u^{-1}\equiv u \,({\rm mod}\,p^m-1)$ if and only if $p^m\equiv 1\,({\rm mod}\,4)$.
\end{lem}
\begin{proof} Since  gcd$(p^m+1,p^m-1)=2$, gcd$(u,p^m-1)$=gcd$(\frac{p^m+1}{2},p^m-1)\leq 2<p$. By Lemma \ref{lem1}, $|C_u|=m$.
If $p^m\equiv 1\,({\rm mod}\,4)$, then
$ u^2=(\frac{p^m+1}{2})^2=(\frac{p^m-1}{4}+1)(p^m-1)+1\equiv 1\,({\rm mod}\,p^m-1).$
Hence $u^{-1}\equiv u\,({\rm mod}\,p^m-1)$ if $p^m\equiv 1\,({\rm mod}\,4)$. If $u^{-1}\equiv u\,({\rm mod}\,p^m-1)$, then $ u^2\equiv 1\,({\rm mod}\,p^m-1).$ This leads to $\frac{(p^m-1)^2}{4}\equiv 0\,({\rm mod}\,p^m-1).$ Therefore $p^m\equiv 1\,({\rm mod}\,4)$.
This completes the proof.
\end{proof}

%

\begin{lem}\label{weight2} Let $p\geq 5$ be an odd prime and $u=\frac{p^m+1}{2}$. Then  $C_p(0,u,v)$ has no codeword of Hamming weight 2 if  and only if $p^m\equiv 1\,({\rm mod}\,4)$ or ${\rm gcd}(v,p^m-1)=1$.
\end{lem}
\begin{proof} The $p$-ary cyclic code $C_p(0,u,v)$ has a codeword of Hamming weight 2 if and only if  there exist two elements $a_1,a_2\in \mathbb{F}_p^*$ and two distinct elements $x_1,x_2\in \mathbb{F}_{p^m}^*$ such that
\begin{equation*}\label{1}
\begin{array}{ll}\left\{\begin{array}{lcr}
a_1+a_2&=&0\\
a_1x_1^u+a_2x_2^u&=&0\\
a_1x_1^v+a_2x_2^v&=&0.
\end{array}\right.
\end{array}
\end{equation*} Then $(\frac{x_1}{x_2})^u=1$ and $(\frac{x_1}{x_2})^v=1$. By Lemma \ref{weight2}, $C_p(0,u,v)$ has no codeword of Hamming weight 2  if and only if $p^m\equiv 1\,({\rm mod}\,4)$ or ${\rm gcd}(v,p^m-1)=1$.
\end{proof}

If $C_p(0,1,u^{-1}v)$  has a codeword of Hamming weight 3, then there exist three elements $a_1,a_2,a_3\in \mathbb{F}_p^*$ and three distinct elements $x_1,x_2,x_3\in \mathbb{F}_{p^m}^*$ such that
\begin{equation}\label{2}
\begin{array}{ll}\left\{\begin{array}{lcr}
a_1+a_2+a_3&=&0\\
a_1x_1+a_2x_2+a_3x_3&=&0\\
a_1x_1^{u^{-1}v}+a_2x_2^{u^{-1}v}+a_3x_3^{u^{-1}v}&=&0.
\end{array}\right.
\end{array}
\end{equation} It is obvious that \eqref{2} can be written as
\begin{equation}\label{3}
\begin{array}{ll}\left\{\begin{array}{lcr}
1+b_1+b_2&=&0\\
x+b_1+b_2y&=&0\\
x^{u^{-1}v}+b_1+b_2y^{u^{-1}v}&=&0,
\end{array}\right.
\end{array}
\end{equation} where $x\neq y$, $x,y\in \mathbb{F}_{p^m}^*\backslash\{1\}$ and $b_1, b_2\in \mathbb{F}_{p}^*\backslash\{-1\}$. Similarly, if $C_p(0,1,uv^{-1})$  has a codeword of Hamming weight 3,
then we have that
\begin{equation}\label{3''}
\begin{array}{ll}\left\{\begin{array}{lcr}
1+b_1+b_2&=&0\\
x+b_1+b_2y&=&0\\
x^{uv^{-1}}+b_1+b_2y^{uv^{-1}}&=&0,
\end{array}\right.
\end{array}
\end{equation} where $x\neq y$, $x,y\in \mathbb{F}_{p^m}^*\backslash\{1\}$ and $b_1, b_2\in \mathbb{F}_{p}^*\backslash\{-1\}$.
Replacing $x$ in \eqref{3} by $x^u$ or replacing $x$ in \eqref{3''} by $x^v$ gives that
\begin{equation}\label{3'}
\begin{array}{ll}\left\{\begin{array}{lcr}
1+b_1+b_2&=&0\\
x^u+b_1+b_2y^u&=&0\\
x^{v}+b_1+b_2y^{v}&=&0.
\end{array}\right.
\end{array}
\end{equation} Hence  $C_p(0,1,u^{-1}v)$ or $C_p(0,1,uv^{-1})$ has no codeword of Hamming weight 3
if and only if  $C_p(0,u,v)$ has no codeword of Hamming weight 3. And $C_p(0,u,v)$ has no codeword of Hamming weight 3 if and only if \eqref{3'} has no distinct solutions $x,y\in \mathbb{F}_{p^m}^*\backslash\{1\}$ for $b_1, b_2\in \mathbb{F}_{p}^*\backslash\{-1\}$.

\section{Four classes of new optimal $p$-ary cyclic codes  with parameters $[p^m-1,p^m-2m-2,4]$}
In this section, we construct four classes  of optimal $p$-ary cyclic codes $C_p(0,1,u^{-1}v)$ or $C_p(0,1,uv^{-1})$ with parameters $[p^m-1,p^m-2m-2,4]$, including three infinite classes  of optimal $p$-ary cyclic codes.   Denote the quadratic character of $\mathbb F_{p^m}$ by $\eta$ which is defined by $\eta(0)=0$, $\eta(x)=1$ if $x$ is a nonzero square in $\mathbb F_{p^m}$ and $\eta(x)=-1$ if $x$ is a nonzero nonsquare in $\mathbb F_{p^m}$.

\subsection{ The first  class of optimal $p$-ary cyclic codes $C_p(0,1,u^{-1}v)$ or $C_p(0,1,uv^{-1})$  with parameters $[p^m-1,p^m-2m-2,4]$}

In this subsection, we consider the exponent $v$ of the form
$v(p^k+1)\equiv p^h+1\,({\rm mod}\,p^m-1)$.
By reducing the solutions of equations in $\mathbb F_{p^m}$ to $\mathbb F_{p}$, an infinite class  of optimal $p$-ary cyclic codes with parameters $[p^m-1,p^m-2m-2,4]$ will be obtained from the exponent $v$ of the form
\begin{equation}\label{4}
v(p^k+1)\equiv p^h+1\,({\rm mod}\,p^m-1),
\end{equation}
where  $0\leq h,k<m$.

\begin{theorem} \label{thm2}  Let $p\geq 5$ be an odd prime, $m$ be a positive integer and $u=\frac{p^m+1}{2}$. Let $h$ and $k$ be  nonnegative integers with ${\rm gcd}(h-k,m)=1$ and ${\rm gcd}(h+k,m)=1$. Let $v$ be  defined by \eqref{4}. Suppose that $p^m\equiv 1 \,({\rm mod}\,4)$ or  ${\rm gcd}( v,p^m-1)=1$. Then the $p$-ary cyclic code $C_p(0,1,u^{-1}v)$ or $C_p(0,1,uv^{-1})$ is optimal and has parameters $[p^m-1,p^m-2m-2,4]$  if and only if the following two systems of equations have no distinct solutions $x,\,y$ in $\mathbb{F}_p^*\backslash \{1\}$ for any $b_1,b_2\in \mathbb{F}_p^*\backslash\{-1\}$:
\begin{equation}\label{weight3}
\begin{array}{ll}(i)\left\{\begin{array}{lcr}
1+b_1+b_2&=&0\\
x+b_1+b_2y&=&0\\
x^v+b_1+b_2y^v&=&0\\
\eta(x)=1,\,\eta(y)&=&1,
\end{array}\right.
\end{array}\quad\quad
\begin{array}{ll}(ii)\left\{\begin{array}{lcr}
1+b_1+b_2&=&0\\
x+b_1-b_2y&=&0\\
x^v+b_1+b_2y^v&=&0\\
\eta(x)=1,\,\eta(y)&=&-1.
\end{array}\right.
\end{array}
\end{equation}
\end{theorem}

\begin{proof}By the theory of multiplies of cyclic codes, $C_p(0,1,u^{-1}v)$ or $C_p(0,1,uv^{-1})$ is equivalent to $C_p(0,u,v)$. Thus the proof is divided into three steps.

\emph{Step 1.} We prove that $v\notin C_u$. Denote $i_m$ the residue of $i$ modulo $m$, where $0\leq i_m< m$.
If $v\in C_u$, then there exists $0\leq i<m$ such that $v\equiv p^iu\,({\rm mod}\,p^m-1)$ by the definition of  cyclotomic coset. Multiplying both sides of this equation with $2(p^k+1)$ yields
\begin{equation}\label{uv}
2p^{(i+h)_m}+2p^i-2p^h-2\equiv0\,({\rm mod}\,p^m-1).
\end{equation} Since $p\geq 5$,  $2p^{(i+h)_m}+2p^i-2p^h-2\leq 4p^{m-1}-4<p^m-4<p^m-1$. This is contrary to \eqref{uv}. Hence $v\notin C_u$.

\emph{Step 2.} We prove that $|C_v|=m$. When $\frac{m}{{\rm gcd}(m,h)}$ is odd, ${\rm gcd}(p^h+1,p^m-1)=2$ by Lemma \ref{gongyinzi}. Thus ${\rm gcd}(v,p^m-1)\leq {\rm gcd}(v(p^k+1),p^m-1)={\rm gcd}(p^h+1,p^m-1)=2<p$. This leads to $|C_v|=m$ by Lemma \ref{lem1}. When $\frac{m}{{\rm gcd}(m,h)}$ is even, suppose that $|C_v|\neq m$ and let $|C_v|=l_v$. This together with Lemma \ref{gong} gives that $l_v\leq \frac{m}{2}$. By the definition of  cyclotomic coset, $v(p^{l_v}-1)\equiv 0\,({\rm mod}\,p^m-1)$. Thus
\begin{equation*}\label{5}\begin{array}{rcl} p^m-1={\rm gcd}(v(p^{l_v}-1),p^m-1)&\mid &
{\rm gcd}(v,p^m-1){\rm gcd}(p^{l_v}-1,p^m-1)\\
&\leq & {\rm gcd}(v(p^k+1),p^m-1){\rm gcd}(p^{l_v}-1,p^m-1)\\
&=& {\rm gcd}(p^h+1,p^m-1)(p^{{\rm gcd}(l_v,m)}-1)\\
&=&(p^{{\rm gcd}(h,m)}+1)(p^{{\rm gcd}(l_v,m)}-1)\\
&\leq& (p^{\frac{m}{2}}+1)(p^{\frac{m}{2}}-1).\\
\end{array}\end{equation*} This implies that  $h=\frac{m}{2}$ and  $l_v=\frac{m}{2}$. Hence ${\gcd}(v(p^k+1),p^m-1)={\gcd}(p^h+1,p^m-1)=p^{\frac{m}{2}}+1$ and $v\equiv 0\,({\rm mod}\,p^\frac{m}{2}+1)$. Thus ${\gcd}(\frac{v}{p^{\frac{m}{2}}+1}(p^k+1),p^m-1)=1.$ This is impossible since $p^\frac{m}{2}+1\mid v$ and ${\gcd}(p^k+1,p^m-1)\geq 2.$ Hence $|C_v|=m$.


\emph{Step 3.} We prove that  $C_p(0,u,v)$ does not have a codeword of Hamming weight 2 or 3. Since $p^m\equiv 1 \,({\rm mod}\,4)$ or ${\rm gcd}( v,p^m-1)=1$, $C_p(0,u,v)$ has no codeword of Hamming weight 2 due to Lemma \ref{weight2}. By the foregoing discussions, $C_p(0,u,v)$ has no codeword of Hamming weight 3 if and only if \eqref{3'} has no distinct solutions $x,\,y\in \mathbb{F}_{p^m}^*\backslash \{1\}$ for any $b_1,b_2\in \mathbb{F}_p^*\backslash\{-1\}$.  Without loss of generality, we will discuss \eqref{3'} by the following two cases due to symmetry.

\emph{Case 1.}  $(\eta(x),\eta(y))=(1,1)$: In this case, \eqref{3'} becomes
\begin{equation}\label{6}
\begin{array}{ll}\left\{\begin{array}{lcr}
1+b_1+b_2&=&0\\
x+b_1+b_2y&=&0\\
x^v+b_1+b_2y^v&=&0,
\end{array}\right.
\end{array}
\end{equation}where $x\neq y$, $x,y\in \mathbb{F}_{p^m}^*\backslash\{1\}$ and $b_1, b_2\in \mathbb{F}_{p}^*\backslash\{-1\}$. From the second and third equations  of \eqref{6}, one has that
\begin{equation}\label{7}
(-(b_1+b_2y))^v=-(b_1+b_2y^v).
\end{equation}Taking $(p^k+1)$-th power on both sides of \eqref{7}, one obtains that $(b_1+b_2y)^{v(p^k+1)}=(b_1+b_2y^v)^{p^k+1}$. Plugging \eqref{4} into it gives that
$y(y^{p^h-v}-1)(y^{v-1}-1)=0$. Therefore $y^{p^h-v}=1$ or $y^{v-1}=1$. Since  ${\rm gcd}(h+k,m)=1$ and ${\rm gcd}(h-k,m)=1$, ${\rm gcd}((p^h-v)(p^k+1), p^m-1)={\rm gcd}(p^{h+k}-1, p^m-1)=p^{{\rm gcd}(h+k,m)}-1=p-1$  and ${\rm gcd}((v-1)(p^k+1), p^m-1)={\rm gcd}(p^{h-k}-1, p^m-1)=p^{{\rm gcd}(h-k,m)}-1=p-1$. Thus
${\rm gcd}(p^h-v, p^m-1)|p-1$ and ${\rm gcd}(v-1, p^m-1)|p-1$. And thus $y^{p-1}=1$, i.e., $y\in \mathbb{F}_p^*$. By the second equation of \eqref{6}, we have $x\in \mathbb{F}_p^*$.

\emph{Case 2.} $(\eta(x),\eta(y))=(1,-1)$: In this case, \eqref{3'} can be written as
\begin{equation}\label{8}
\begin{array}{ll}\left\{\begin{array}{lcr}
1+b_1+b_2&=&0\\
x+b_1-b_2y&=&0\\
x^v+b_1+b_2y^v&=&0,
\end{array}\right.
\end{array}
\end{equation}where $x\neq y$, $x,y\in \mathbb{F}_{p^m}^*\backslash\{1\}$ and $b_1, b_2\in \mathbb{F}_{p}^*\backslash\{-1\}$. Similar as in case 1, in this case \eqref{8} is reduced to $y(y^{p^h-v}+1)(y^{v-1}+1)=0$. Thus $y^{2(p^h-v)}=1$ or $y^{2(v-1)}=1$. Since ${\rm gcd}(h+k,m)=1$ and  ${\rm gcd}(h-k,m)=1$, ${\rm gcd}(2(p^h-v)\cdot \frac{p^k+1}{2},p^m-1)={\rm gcd}(p^{h+k}-1,p^m-1)=p-1$ and ${\rm gcd}(2(v-1)\cdot \frac{p^k+1}{2},p^m-1)={\rm gcd}(p^{h-k}-1,p^m-1)=p-1$.
Hence $y^{p-1}=1$, i.e., $y\in \mathbb{F}_p^*.$ By the second equation of \eqref{8}, $x\in \mathbb{F}_p^*$.

To sum up, $C_p(0,1,u^{-1}v)$ or $C_p(0,1,uv^{-1})$ has no codeword of Hamming weight 3 if and only if \eqref{weight3} has no distinct solutions $x,\,y\in \mathbb{F}_p^*\backslash \{1\}$ for any $b_1, b_2\in \mathbb{F}_p^*\backslash\{-1\}$. Then the desired result follows immediately. This completes the proof.
\end{proof}


\begin{cor}\label{cor2} Let $p=5$, $m$ be a positive integer and $u=\frac{p^m+1}{2}$. Let $v$ satisfy \eqref{4} with ${\rm gcd}(h-k,m)=1$ and ${\rm gcd}(h+k,m)=1$. Then $v\equiv 1,\frac{p+1}{2}\,({\rm mod}\,p-1)$. The $p$-ary cyclic code $C_p(0,1,u^{-1}v)$ or $C_p(0,1,uv^{-1})$  has parameters $[p^m-1,p^m-2m-2,4]$ if

\noindent 1) $v\equiv 1\,({\rm mod}\,p-1)$ and $m$ is odd; or

\noindent 2) $v\equiv 3\,({\rm mod}\,p-1)$ and $m$ is even.
\end{cor}

\begin{proof} We first prove $v\equiv 1,\frac{p+1}{2}\,({\rm mod}\,p-1)$. Since $v(p^k+1)\equiv p^h+1\,({\rm mod}\,p^m-1)$, $v(p^k+1)\equiv p^h+1\,({\rm mod}\,p-1)$. This leads to $v(p^k-1)-(p^h-1)+2(v-1)\equiv 0\,({\rm mod}\,p-1)$.
Thus  $2v\equiv 2\,({\rm mod}\,p-1)$. And thus  $v\equiv 1,\frac{p+1}{2}\,({\rm mod}\,p-1)$.
By Theorem \ref{thm2}, one needs to prove that \eqref{weight3} has no distinct solutions $x,y$ in $\mathbb{F}_5^*\backslash \{1\}$ for any $b_1,b_2\in \mathbb{F}_5^*\backslash\{-1\}$.

\emph{Case 1.} $v\equiv 1\,({\rm mod}\,4)$ and $m$ is odd: In this case, $\eta(\pm 1)=1$ and $\eta(\pm 2)=-1$. Thus the system (i) of \eqref{weight3} has solutions $x=-1$ and $y=-1$ in $\mathbb{F}_5^*\backslash \{1\}$ since $\eta(x)=\eta(y)=1$. This is contrary to that $x\neq y$. As a result, the system (i) of \eqref{weight3} has no distinct solutions $x,y\in\mathbb{F}_5^*\backslash \{1\}$ for any $b_1,b_2\in \mathbb{F}_5^*\backslash\{-1\}$.
Since $v\equiv 1\,({\rm mod}\,4)$ and $x,y\in\mathbb{F}_5^*\backslash \{1\}$, $x^v=x$ and $y^v=y$. The system (ii) of \eqref{weight3} becomes
\begin{equation}\label{110}\begin{array}{ll}\left\{\begin{array}{lcr}
1+b_1+b_2&=&0\\
x+b_1-b_2y&=&0\\
x+b_1+b_2y&=&0\\
\eta(x)=1,\,\eta(y)&=&-1.
\end{array}\right.
\end{array}\end{equation} From the second and third equations of \eqref{110} we have that  $b_2y=0$.  This is contrary to that $b_2\neq 0$ and $y\neq 0$. Hence the system (ii) of \eqref{weight3} has no distinct solutions $x,y\in\mathbb{F}_5^*\backslash \{1\}$ for any $b_1,b_2\in \mathbb{F}_5^*\backslash\{-1\}$.

\emph{Case 2.} $v\equiv 3\,({\rm mod}\,4)$ and $m$ is even: In this case, all elements in $\mathbb{F}_5^*$ are squares since $\mathbb{F}_5^*=\{\alpha^{\frac{5^m-1}{5-1}i}: i=0,1,2,3\}$ and $\frac{5^m-1}{5-1}$ is even. Therefore the  system (ii) of \eqref{weight3} has no solution $y\in\mathbb{F}_5^*$ such that $\eta(y)=-1$. The system (i) of \eqref{weight3} can be written as
\begin{equation*}\label{111}\begin{array}{ll}\left\{\begin{array}{lcr}
1+b_1+b_2&=&0\\
x+b_1+b_2y&=&0\\
x^3+b_1+b_2y^3&=&0\\
\eta(x)=1,\,\eta(y)&=&1,
\end{array}\right.
\end{array}\end{equation*}  which can be reduced to
\begin{equation}\label{12'}\begin{array}{ll}\left\{\begin{array}{lcr}
(x-1)+b_2(y-1)&=&0\\
(x^3-x)+b_2(y^3-y)&=&0.
\end{array}\right.
\end{array}\end{equation} The first equation of \eqref{12'} gives that $b_2=-\frac{x-1}{y-1}$. Replacing $b_2$ in the second equation of \eqref{12'} by $-\frac{x-1}{y-1}$ yields that $x+y=4$. Hence $x=y=2$ since $x,y\in \mathbb{F}_5^*\backslash \{1\}$. This is contrary to that $x\neq y$. As a result, the system (i) of \eqref{weight3} has no distinct solutions $x,y\in\mathbb{F}_5^*\backslash \{1\}$ for any $b_1,b_2\in \mathbb{F}_5^*\backslash\{-1\}$.
Then the desired result follows immediately. This completes the proof.
\end{proof}

\begin{cor}\label{cor3} Let $p\geq7$ be an odd prime and $m$ be a positive integer. Let $u=\frac{p^m+1}{2}$ and $v$ be defined by \eqref{4}. Then  the $p$-ary cyclic code $C_p(0,1,u^{-1}v)$ or $C_p(0,1,uv^{-1})$ has parameters $[p^m-1,p^m-2m-2,3]$.
\end{cor}

\begin{proof} By Corollary \ref{cor2}, $v\equiv 1\,({\rm mod}\,p-1)$ or $v\equiv \frac{p+1}{2}\,({\rm mod}\,p-1)$. We claim that  the system (i) of \eqref{weight3} has two distinct  solutions  in $\mathbb{F}_p^*\backslash \{1\}$.
Let $S=\{z\in \mathbb F_p^*: \eta(z)=1, z\neq 1\}$. Since $\mathbb{F}_p^*=\{\alpha^{\frac{p^m-1}{p-1}i}: 0\leq i\leq p-2\}$ and $p\geq 7$,
$|S|=\frac{p-1}{2}-1\geq 2$ for odd $m$ and $|S|=p-2\geq 5$ for even $m$. Then there exist two distinct elements $x,y \in S$.  If $v\equiv 1\,({\rm mod}\,p-1)$, then $x^v=x$. If $v\equiv \frac{p+1}{2}\,({\rm mod}\,p-1)$, then $x^v=\eta(x)x=x$.
Hence
the system (i) of \eqref{weight3} becomes
\begin{equation}\label{v=1}
\begin{array}{ll}\left\{\begin{array}{lcr}
b_1+b_2&=&-1\\
b_1+b_2y&=&-x,
\end{array}\right.
\end{array}
\end{equation} where $x,y\in\mathbb{F}_p^*\backslash \{1\}$ and $\eta(x)=\eta(y)=1$.
Hence \eqref{v=1} has  solutions $b_1=\frac{x-y}{y-1}$, $b_2=\frac{1-x}{y-1}$ in $\mathbb F_p^*\backslash \{-1\}$.
Therefore
 there exist two elements $b_1,b_2\in \mathbb F_p^*\backslash \{-1\}$ such that the system (i) of \eqref{weight3} has two distinct solutions  $x,y \in\mathbb F_p^*\backslash \{1\}$. Then the desired result follows immediately from Theorem \ref{thm2}.
This completes the proof.
\end{proof}


%
\begin{example} Let  $m=4$,  $h=0$ and $k=1$. Let $u=\frac{5^m+1}{2}$ and $v(5^k+1)\equiv 5^h+1\,({\rm mod}\,5^m-1)$. Then $u=313$ and $v=315$. Thus $u^{-1}v=3$.  Let $\alpha$ be the generator of $\mathbb F_{5^m}^*$ with $\alpha^4+4\alpha^2+4\alpha+2=0$. Then $C_5(0,1,3)$ has parameters $[624,615,4]$  and generator polynomial $x^9 + x^8 + 2x^5 + 2x^3 + 3x^2 + 2x + 4$.
\end{example}

\begin{example} Let  $m=5$, $h=1$ and $k=2$. Let $u=\frac{5^m+1}{2}$ and $v(5^k+1)\equiv 5^h+1\,({\rm mod}\,5^m-1)$. Then $u=1563$ and $v=525$. Thus $u^{-1}v=2087$ and $uv^{-1}=2163$.
Let $\alpha$ be the generator of $\mathbb F_{5^m}^*$ with $\alpha^5+4\alpha+3=0$. Then $C_5(0,1,u^{-1}v)$ or $C_5(0,1,uv^{-1})$ has parameters $[3124,3113,4]$  and generator polynomial $x^{11} + x^{10} + 3x^9 + 2x^8 + x^7 + 3x^5 + 2x^4 + 4x^3 + 3x^2 + x + 4$ or $x^{11} + 2x^{10} + 2x^9 + 2x^8 +x^7+ 4x^6 + 2x^5 + 4x^4 + 4x^3 + 3x^2 + x + 4$.
\end{example}

\begin{remark}\label{remark1} In 2020, Liu and Cao \cite{Liu1} proved that $C_5(0,1,v)$ has parameters $[5^m-1,5^m-2m-2,4]$, where $v(5^k+1)\equiv 5^h+1\,({\rm mod}\,5^m-1)$ with $v\equiv 3\,({\rm mod}\,4)$ and ${\rm gcd}(m,h+k)={\rm gcd}(m,h-k)=1$. This is the special case $p=5$ of Theorem \ref{thm2}.
\end{remark}

%

\subsection{ The second class of optimal $p$-ary cyclic codes  $C_p(0,1,u^{-1}v)$ or $C_p(0,1,uv^{-1})$ with parameters $[p^m-1,p^m-2m-2,4]$}
In this subsection,
by  reducing the solutions of equations in $\mathbb F_{p^m}$ to $\mathbb F_{p}$, an infinite class of optimal $p$-ary cyclic codes with parameters $[p^m-1,p^m-2m-2,4]$ will be obtained from the exponent $v$ of the form
\begin{equation}\label{14}
v(p^k-1)\equiv p^h-1\,({\rm mod}\,p^m-1), \,\,0\leq k,h<m.
\end{equation}
The
following theorem can be proved by the same approach as the proof of Theorem \ref{thm2}. So we omit the proof of the following theorem.
\begin{theorem} \label{thm3}  Let $p\geq 5$ be an odd prime, $m$ be an odd integer and $u=\frac{p^m+1}{2}$. Let $v$  be  defined by \eqref{14} such that  ${\rm gcd}(h,m)=1$ and ${\rm gcd}(h-k,m)=1$. Let $p^m\equiv 1 \,({\rm mod}\,4)$ or  ${\rm gcd}( v,p^m-1)=1$. Then the $p$-ary cyclic code $C_p(0,1,u^{-1}v)$ or $C_p(0,1,uv^{-1})$ has parameters $[p^m-1,p^m-2m-2,4]$  if and only if \eqref{weight3} has no distinct solutions $x,y\in\mathbb{F}_p^*\backslash \{1\}$ for  $b_1,b_2\in \mathbb{F}_p^*\backslash\{-1\}$.
\end{theorem}

\begin{remark}\label{remark2} Two remarks about Theorem \ref{thm3} are the following:

\noindent 1)  If $m$ is even, then $h$ is odd and $k$ is even since ${\rm gcd}(h,m)={\rm gcd}(h-k,m)=1$. And then $p^m-1\equiv 0\,({\rm mod}\,2p-2)$ and $p^k-1\equiv 0\,({\rm mod}\,2p-2)$. This together with \eqref{14} leads to  $p^h-1\equiv 0\,({\rm mod}\,2p-2)$. This is impossible. This is the reason why let $m$ be odd in Theorem \ref{thm3}.

\noindent 2) In general, it is easier to determine the solution of \eqref{weight3} in $\mathbb{F}_p^*$ than in $\mathbb{F}_{p^m}^*$. We give two corollaries for the cases $p=5$ and $p=7$ as follows.
\end{remark}

\begin{cor}\label{cor4} Let $p=5$ and $m$ be odd. Let $u=\frac{p^m+1}{2}$ and  $v$ be defined by Theorem \ref{thm3}. The $p$-ary cyclic code $C_p(0,1,u^{-1}v)$ or $C_p(0,1,uv^{-1})$  has parameters $[p^m-1,p^m-2m-2,4]$ if and only if
$v\equiv 1\,({\rm mod}\,4)$ or
$v\equiv 2\,({\rm mod}\,4)$.
\end{cor}
\begin{proof} 
By Theorem \ref{thm3}, one needs to prove that \eqref{weight3} has no distinct solutions $x,y \in\mathbb F_5^*\backslash\{1\}$ for $b_1,b_2\in\mathbb F_5^*\backslash\{-1\}$ if and only if $v\equiv 1\,({\rm mod}\,4)$ or
$v\equiv 2\,({\rm mod}\,4)$. Since  $m$ is odd, $\eta(\pm1)=1$ and $\eta(\pm2)=-1$. If the system (i) of \eqref{weight3} has two solutions $x,y\in\mathbb{F}_5^*\backslash \{1\}$, then $x=y=-1$ since $\eta(x)=\eta(y)=1$. This is contrary to that $x\neq y$. Hence  the  system (i) of \eqref{weight3} has no distinct solutions $x,y\in\mathbb{F}_5^*\backslash \{1\}$ for $b_1,b_2\in \mathbb{F}_5^*\backslash\{-1\}$. We will discuss the solutions $x,y \in\mathbb F_5^*\backslash\{1\}$ of the system (ii) of  \eqref{weight3} by distinguishing among the following cases:

\emph{Case 1:} $v\equiv 0\,({\rm mod}\,4)$: In this case  $x^v=1$ and $y^v=1$. It is easy to check that the  system (ii) of \eqref{weight3} has a solution $x=-1$, $y=2$, $b_1=3$ and $b_2=1$.

\emph{Case 2:}  $v\equiv 1\,({\rm mod}\,4)$: In this case $x^v=x$ and $y^v=y$.
The second and third equations of the system (ii) of \eqref{weight3} is simplified to $b_2y=0$. Thus $b_2=0$ or $y=0$. This is contrary to that
$b_2\neq0$ and $y\neq0$. Hence the  system (ii) of \eqref{weight3} has no distinct solutions $x,y\in\mathbb{F}_5^*\backslash \{1\}$ for $b_1,b_2\in \mathbb{F}_5^*\backslash\{-1\}$.

\emph{Case 3:}  $v\equiv 2\,({\rm mod}\,4)$:
In this case  the  system (ii) of \eqref{weight3} is reduced to $(x-1)(x+y)=0$, which implies that $x=-y$ since $x\neq 1$. Thus $\eta (y)=\eta(-y)=\eta(x)=1.$ This is contrary to that $\eta(y)=-1$. Therefore the  system (ii) of \eqref{weight3} has no distinct solutions $x,y\in\mathbb{F}_5^*\backslash \{1\}$ for $b_1,b_2\in \mathbb{F}_5^*\backslash\{-1\}$.

\emph{Case 4:} $v\equiv 3\,({\rm mod}\,4)$: In this case   $x^v=x^3$ and $y^v=y^3$.
It is straightforward to check that the  system (ii) of  \eqref{weight3} has a solution $x=-1$, $y=2$, $b_1=3$ and $b_2=1$. This completes the proof.
\end{proof}

\begin{cor}\label{cor5} Let $p=7$ and  $m$ be odd. Let $u=\frac{p^m+1}{2}$ and $v$ be defined by Theorem \ref{thm3} with ${\rm gcd}(v,p^m-1)=1$. Then the $p$-ary cyclic code $C_p(0,1,u^{-1}v)$ or $C_p(0,1,uv^{-1})$ has parameters $[p^m-1,p^m-2m-2,4]$ if and only if $v\equiv 5\,({\rm mod}\,6)$.
\end{cor}
\begin{proof} 
According to Theorem \ref{thm3},  it suffices to show that \eqref{weight3} has no distinct solutions $x,y\in\mathbb{F}_7^*\backslash \{1\}$ for $b_1,b_2\in\mathbb{F}_7^*\backslash\{-1\}$ if and only if $v\equiv 5\,({\rm mod}\,6)$. Since $m$ is odd, $\eta(1)=\eta(2)=\eta(4)=1$ and $\eta(3)=\eta(5)=\eta(-1)=-1$. Note that $v$ is odd. We only need to consider  three cases $v\equiv 1\,({\rm mod}\,6)$, $v\equiv 3\,({\rm mod}\,6)$ and $v\equiv 5\,({\rm mod}\,6)$.

When $v\equiv 1\,({\rm mod}\,6)$ or $v\equiv 3\,({\rm mod}\,6)$, it is straightforward to check  that the system (i) of \eqref{weight3} has a solution $x=2$, $y=4$, $b_1=4$ and $b_2=2$.

When $v\equiv 5\,({\rm mod}\,6)$, let $x,y\in\mathbb F_7^*\backslash \{1\}$ such that $\eta(x)=\eta(y)=1$. Then the system (i) of \eqref{weight3} is reduced to
\begin{equation}\label{15}   \frac{x^5-1}{x-1}-\frac{y^5-1}{y-1}=0,\end{equation}
where $x\neq y$, $x,y\in\mathbb F_7^*\backslash \{1\}$ and $\eta(x)=\eta(y)=1$. Thus
$(x,y)=(2,4)$ or $(4,2)$. This is contrary to that $(x,y)=(2,4)$ or $(4,2)$ is not a solution of \eqref{15}. Hence the system (i) of \eqref{weight3}
has no distinct solutions $x,y\in\mathbb F_7^*\backslash \{1\}$ for $b_1,b_2\in\mathbb F_7^*\backslash \{-1\}$.
If $y=-1$, it then follows from the system (ii) of \eqref{weight3} that $x=1$. This is impossible since $x\in\mathbb F_7^*\backslash \{1\}$. Thus $y\neq -1$. And thus
the system (ii) of \eqref{weight3} is simplified to
\begin{equation}\label{16}\frac{x^5-1}{x-1} +\frac{y^5-1}{y+1}=0,\end{equation}
where $x\neq y$, $x,y\in\mathbb F_7^*\backslash \{1\}$, $\eta(x)=1$ and $\eta(y)=-1$.
Hence  $x\in\{2,4\}$ and $y\in\{3,5\}$. This is contrary to that $(x,y)\in \{2,4\}\times \{3,5\}$ are not solutions of \eqref{16}.
This completes the proof.
\end{proof}

\begin{remark}\label{remark3} 
In 2020, Liu and Cao \cite{Liu1} proved that $C_5(0,1,v)$ has parameters $[5^m-1,5^m-2m-2,4]$, where $v(5^k-1)\equiv 5^h-1\,({\rm mod}\,5^m-1)$, $v\equiv 2\,({\rm mod}\,4)$ or $v\equiv 0\,({\rm mod}\,4)$, and ${\rm gcd}(m,h)={\rm gcd}(m,k)={\rm gcd}(m,h-k)=1$.  This is the special case $p=5$ of Theorem \ref{thm3}  if $v\equiv 2\,({\rm mod}\,4)$.
\end{remark}

\begin{example} Two examples of the codes of Theorem \ref{thm3} are the following:

\noindent 1) Let $p=5$ and $m=3$. Let $u=\frac{5^m+1}{2}=63$ and  $\alpha$ be the generator of $\mathbb F_{5^m}^*$ with $\alpha^3+3\alpha+3=0$. If  $v=\frac{5^m-1}{4}+\frac{5^2-1}{4}=37$, then $u^{-1}v=99$ and $uv^{-1}=119$. The code  $C_5(0,1,u^{-1}v)$ or $C_5(0,1,uv^{-1})$ has parameters $[124,117,4]$  and generator polynomial $x^7 + 2x^5 + 3x^2 + 4$. If $v=\frac{5(5^m-1)+(5-1)}{24}=26$, then $u^{-1}v=26$.
The  code $C_5(0,1,u^{-1}v)$  has parameters $[124,117,4]$  and generator polynomial $x^7 + x^6 + x^5 + 2x^4 + x^3 + 2x^2 + 2$.

\noindent 2) Let $p=7$ and  $m=3$. Let $u=\frac{7^m+1}{2}=172$ and $\alpha$ be the generator of $\mathbb F_{7^m}^*$ with $\alpha^3+6\alpha+4=0$. Let $6v\equiv 48\,({\rm mod}\,7^m-1)$. Then $v=65,179,293$ and $uv^{-1}=50, 278,164$. The codes $C_7(0,1,uv^{-1})$   have parameters $[342,335,4]$  and generator polynomials $x^7 + 5x^6 + 4x^5 + 3x^4 + 3x^3 + 3x^2 + x + 1$, $x^7 + 5x^6 + 4x^4 + 6x^3 + x^2 + 3x + 1$ and $x^7 + 5x^6 + 6x^5 + 6x^4 + 5x^3 + 4x^2 + 1$.
\end{example}
%
%

\subsection{ The third class of optimal $p$-ary cyclic codes  $C_p(0,1,u^{-1}v)$ or $C_p(0,1,uv^{-1})$ with parameters $[p^m-1,p^m-2m-2,4]$}\label{subsection3}

In this subsection, new optimal $p$-ary cyclic codes $C_p(0,1,u^{-1}v)$ or $C_p(0,1,uv^{-1})$  will be obtained from the exponent $v$ of the form
$$v=\frac{p^m-1}{2}+r,\,\, 1\leq r\leq p^m-2.$$

\begin{lem}\label{lem10} Let $p\geq 5$ be an odd prime and $m$ be a positive integer. Suppose $u=\frac{p^m+1}{2}$ and $k<m$ is a nonnegative integer. Then $v\notin C_u$ and $|C_v|=m$ if

\noindent 1) $v=\frac{p^m-1}{2}+p^k+1$ and $\frac{m}{{\rm gcd}(k,m)}\equiv 0\,({\rm mod}\,4)$; or

\noindent 2) $v=\frac{p^m-1}{2}+2p^k$ and $m\equiv 0\,({\rm mod}\,4)$; or

\noindent 3) $v=\frac{p^m-1}{2}+\frac{p^m-3}{2}$ and $m\equiv 0\,({\rm mod}\,2)$; or

\noindent 4) $v=\frac{p^m-1}{2}-1$.
\end{lem}

\begin{proof} We only give the proof of  1) since  2), 3) and 4) can be proved in the same manner. We first prove that $v\notin C_u$. If $v\in C_u$, then there exists  $0\leq i<m$ such that $v\equiv p^iu\,({\rm mod}\,p^m-1)$. It is reduced to $p^k+1\equiv p^i\,({\rm mod}\,\frac{p^m-1}{2})$. Thus $\frac{p^m-1}{2}\mid p^k-p^i+1$. This is impossible since $p^k-p^i+1\leq p^{m-1}<\frac{p^m-1}{2}$. Hence $v\not\in C_u$.

It is now time to prove that $\mid C_v\mid =m$. By Lemma \ref{gongyinzi}, ${\rm gcd}(v,p^m-1)={\rm gcd}(\frac{p^m-1}{2}+p^k+1,\frac{p^m-1}{2})={\rm gcd}(p^k+1,\frac{p^m-1}{2})=p^{{\rm gcd}(k,m)}+1$ since $\frac{m}{{\rm gcd}(k,m)}\equiv 0\,({\rm mod}\,4)$.
Let $|C_v|=l_v$. By the definition of cyclotomic coset,  $v(p^{l_v}-1)\equiv 0\,({\rm mod}\,p^m-1)$. Thus
$(p^{{\rm gcd}(k,m)}+1)(p^{l_v}-1)\equiv 0\,({\rm mod}\,p^m-1)$. If $l_v\neq m$, then $l_v\leq \frac{m}{2}$ due to Lemma \ref{gong}. This together with the fact that $\frac{m}{{\rm gcd}(k,m)}\equiv 0\,({\rm mod}\,4)$ gives that  $(p^{{\rm gcd}(k,m)}+1)(p^{l_v}-1)\leq (p^{\frac{m}{4}}+1)(p^{\frac{m}{2}}-1)<p^m-1$. This is contrary to that $(p^{{\rm gcd}(k,m)}+1)(p^{l_v}-1)\equiv 0\,({\rm mod}\,p^m-1)$. Therefore $l_v= m$. This completes the proof.
\end{proof}

\begin{theorem}\label{thm4} Let $p\geq 5$ be an odd prime, $m$ be a positive integer  and $u=\frac{p^m+1}{2}$. Suppose $k<m$.
Then $C_p(0,1,u^{-1}v)$ or $C_p(0,1,uv^{-1})$  is optimal and has parameters $[p^m-1,p^m-2m-2,4]$ if

\noindent 1) $v=\frac{p^m-1}{2}+p^k+1$ and $\frac{m}{{\rm gcd}(k,m)}\equiv 0\,({\rm mod}\,4)$; or

\noindent 2) $v=\frac{p^m-1}{2}+2p^k$ and $m\equiv 0\,({\rm mod}\,4)$; or

\noindent 3) $v=\frac{p^m-1}{2}+\frac{p^m-3}{2}$ and $m\equiv 0\,({\rm mod}\,2)$; or

\noindent 4) $v=\frac{p^m-1}{2}-1$, $p^m\equiv 1\,({\rm mod}\,4)$ or ${\rm gcd}(v,p^m-1)=1$.

\end{theorem}
\begin{proof}By the theory of multiplies of cyclic codes, $C_p(0,1,u^{-1}v)$ or $C_p(0,1,uv^{-1})$ is equivalent to $C_p(0,u,v)$. Hence we only need to prove that $C_p(0,u,v)$ is optimal and has parameters $[p^m-1,p^m-2m-2,4]$.

According to Lemma \ref{lem10}, $v\notin C_u$ and $|C_v|=m$. By Lemma \ref{lem2}, the dimension of
$C_p(0,u,v)$ is equal to $p^m-2m-2$. Since $m$ is even in  1), 2) and 3),
$p^m\equiv 1\,({\rm mod}\,4)$. By Lemma \ref{weight2}, $C_p(0,u,v)$ has no codeword of Hamming weight 2. Suppose that
$C_p(0,u,v)$ has a codeword of Hamming weight 3. Then there exist two elements $b_1,b_2\in \mathbb F_p^*\backslash \{-1\}$ and two distinct elements $x,y \in \mathbb F_{p^m}^*\backslash \{1\}$ such that \eqref{3'} holds.

We first consider  1).
When $(\eta(x),\eta(y))=(1,1)$,
\eqref{3'} becomes
\begin{equation}\label{17}
\begin{array}{ll}\left\{\begin{array}{lcr}
1+b_1+b_2&=&0\\
x+b_1+b_2y&=&0\\
x^{p^k+1}+b_1+b_2y^{p^k+1}&=&0.
\end{array}\right.
\end{array}
\end{equation}From the second equation of \eqref{17}, we get that $x=-(b_1+b_2y)$. Replacing $x$ by $-(b_1+b_2y)$ in the third equation of \eqref{17} gives that
$$\begin{array}{rcl}&&x^{p^k+1}+b_1+b_2y^{p^k+1}\\
&&=(b_1+b_2y)^{p^k+1}+b_1+b_2y^{p^k+1}\\
&&=(b_1+b_2y)(b_1+b_2y^{p^k})+b_1+b_2y^{p^k+1}\\
&&= (b_2+b_2^2)y^{p^k+1}+b_1b_2(y^{p^k}+y)+(b_1+b_1^2)\\
&&=-b_1b_2(y^{p^k+1}-y^{p^k}-y+1)\\
&&=-b_1b_2(y-1)^{p^k+1}\\
&&=0.
\end{array}$$ Hence
$(y-1)^{p^k+1}=0$ due to $b_1b_2\neq 0$. Thus $y=1$. This is contrary to that $y\neq 1$.
When $(\eta(x),\eta(y))=(1,-1)$, \eqref{3'} can be written as
\begin{equation}\label{18}
\begin{array}{ll}\left\{\begin{array}{lcr}
1+b_1+b_2&=&0\\
x+b_1-b_2y&=&0\\
x^{p^k+1}+b_1-b_2y^{p^k+1}&=&0.
\end{array}\right.
\end{array}
\end{equation}
Substituting the first and second equations into the third equation of \eqref{18} yields that
$$\begin{array}{rcl}&& x^{p^k+1}+b_1-b_2y^{p^k+1}\\
&&=(b_2y-b_1)^{p^k+1}+b_1-b_2y^{p^k+1}\\
&&= (b_2y-b_1)(b_2y^{p^k}-b_1)+b_1-b_2y^{p^k+1}\\
&&=(b_2^2-b_2)y^{p^k+1}-b_1b_2(y^{p^k}+y)+(b_1^2+b_1)\\
&&=-b_1b_2(y+1)^{p^k+1}-2b_2y^{p^k+1}\\
&&=0.
\end{array}$$ This is reduced to
\begin{equation}\label{19}
(1+\frac{1}{y})^{p^k+1}=-\frac{2}{b_1}.
\end{equation}Taking the $(p-1)$-th power on both sides of \eqref{19} yields that
\begin{equation}\label{20}
(1+\frac{1}{y})^{(p-1)(p^k+1)}=1.
\end{equation}Since $\frac{m}{{\rm gcd}(k,m)}\equiv 0\,({\rm mod}\,4)$, ${\rm gcd}((p-1)(p^k+1),p^m-1)=(p-1){\rm gcd}(p^k+1,\frac{p^m-1}{p-1})=(p-1)(p^{{\rm gcd}(k,m)}+1)$ due to Lemma \ref{gongyinzi}. Thus \eqref{20} becomes $(1+\frac{1}{y})^{(p-1)(p^{{\rm gcd}(k,m)}+1)}=1$. And thus
\begin{equation}\label{21}
(1+\frac{1}{y})^{p^{2{\rm gcd}(k,m)}-1}=1.
\end{equation}Hence $y\in \mathbb F_{p^{2{\rm gcd}(k,m)}}^*$.
Since $\mathbb F_{p^{2{\rm gcd}(k,m)}}^*=\{\alpha^{\frac{p^m-1}{p^{2{\rm gcd}(k,m)}-1}i},\, 0\leq i<p^{2{\rm gcd}(k,m)}-1\}$ and $\frac{m}{{\rm gcd}(k,m)}\equiv 0\,({\rm mod}\,4)$, all elements of $\mathbb F_{p^{2{\rm gcd}(k,m)}}^*$ are squares. Hence \eqref{21} has no solution $y\in\mathbb F_{p^{2{\rm gcd}(k,m)}}^*$ such that $\eta (y)=-1$. Therefore  1) has no codeword of Hamming weight 3.  2)  can be proved by the same approach as the proof of  1).

We now prove  3). If $(\eta(x),\eta(y))=(1,1)$, then $x=-(b_1+b_2y)$ from the second equation of \eqref{3'}. Plugging $x=-(b_1+b_2y)$ into the third equation of  \eqref{3'} gives that $b_1+b_2y^{-1}-(b_1+b_2y)^{-1}=0$. Multiplying both sides of it by $y(b_1+b_2y)$ yields that $(y-1)^2=0$. Thus $y=1$, which is contrary to that $y\neq 1$. If $(\eta(x),\eta(y))=(1,-1)$, one has similarly that
\begin{equation}\label{1111}
y^2-\frac{2}{b_1}y-1=0.
\end{equation}Since $m$ is even, all elements of $\mathbb{F}_p^*$ are squares. This implies that $1+b_1^2$ is a square in $\mathbb{F}_p^*$. Hence \eqref{1111} has two solutions $y=\frac{1\pm \sqrt{1+b_1^2}}{b_1}$ in $ \mathbb{F}_p^*$, which indicates $\eta(y)=1$. This is contrary to that $\eta(y)=-1$. Hence $C_p(0,u,v)$ has no codeword of Hamming weight 3.

It is  now time to prove that  4). If $(\eta(x),\eta(y))=(1,1)$, then \eqref{3'} is reduced to $(y-1)^2=0.$ Thus $y=1$. This is contrary to that $y\neq 1$. If $(\eta(x),\eta(y))=(1,-1)$, then \eqref{3'} is simplified to $(y+1)^2=0.$ Thus $y=-1$. This together with the first equation of \eqref{3'} yields that $x=b_2y-b_1=-b_2-b_1=1$. This is contrary to that $x\neq 1$. Hence $C_p(0,u,v)$ has no codeword of Hamming weight 3. This completes the proof.
\end{proof}

\begin{remark}\label{remark4}
In 2016, Xu and Cao \cite{XCX} proved that $C_5(0,1,e)$ has parameters $[5^m-1,5^m-2m-2,4]$, where $m$ is even and $e=\frac{5^m-1}{2}-1$. Let $v$ be given by 3) of Theorem \ref{thm4}. Then  $u^{-1}v\equiv e\,({\rm mod}\,5^m-1)$. Hence  $C_5(0,1,e)$ is the special case $p=5$ of Theorem \ref{thm4}.
\end{remark}

\begin{example} Let $m=4$ and $u=\frac{5^m+1}{2}=313$. Let $\alpha$ be the generator of $\mathbb F_{5^m}^*$ with $\alpha^4 + 4\alpha^2 + 4\alpha + 2=0$.
Four examples of the codes of Theorem \ref{thm4} are the following:
\begin{description}
  \item[1)] Let $k=1$ and $v= \frac{5^m-1}{2}+5^k+1=318$. Then $u^{-1}v=318$. The code $C_5(0,1,u^{-1}v)$ has parameters $[624,615,4]$ and generator polynomial $x^9 + 2x^8 + 4x^7 + 3x^5 + 2x^4 + 4x^3 + 2x^2 + 2$.
  \item[2)] Let $k=1$ and $v= \frac{5^m-1}{2}+2\cdot5^k+1=322$. Then $u^{-1}v=322$. The code $C_5(0,1,u^{-1}v)$ has parameters $[624,615,4]$ and generator polynomial $x^9 + x^8 + 2x^7 + 3x^6 + 2x^5 + 2x^3 + 2x + 2$.
  \item[3)] Let $v= \frac{5^m-1}{2}+\frac{5^m-3}{2}=623$. Then $v^{-1}=623$ and $u^{-1}v=uv^{-1}=uv=311$. The code $C_5(0,1,uv)$ has parameters $[624,615,4]$ and generator polynomial $x^9 + 2x^8 + 3x^7 + 4x^5 + 4x^4 + 2x^3 + x^2 + 4x + 4$.
\item[4)] Let $v= \frac{5^m-1}{2}-1=311$. Then $v^{-1}=311$ and $u^{-1}v=uv^{-1}=uv=623$. The code $C_5(0,1,uv)$ has  parameters $[624,615,4]$ and generator polynomial $x^9 + x^8 + 4x^7 + x^6 + 4x^5 + x^4 + 4x^3 + x^2 + 4x + 4$.
\end{description}
\end{example}
\subsection{ The fourth  class of optimal $p$-ary cyclic codes  $C_p(0,1,u^{-1}v)$ or $C_p(0,1,uv^{-1})$ with parameters $[p^m-1,p^m-2m-2,4]$}

In this subsection, we consider the exponent $v$ of the form
\begin{equation}\label{v=-2} v\equiv -2\,({\rm mod}\,p^m-1).
\end{equation}
The following lemma is useful in the sequel.

\begin{lem}\label{11}$($\cite[Theorems 1 and 2 on page 53]{Ireland}$)$ Let $p$ be an odd prime as well as $q$.

\noindent 1) If $q\equiv 1\,({\rm mod}\,4)$, then $\eta(q)=1$ in $\mathbb F_p$  if and only if $p\equiv r \,({\rm mod}\,q)$, where $\eta(r)=1$ in $\mathbb F_q$;

\noindent 2) If $q\equiv 3\,({\rm mod}\,4)$, then $\eta(q)=1$ in $\mathbb F_p$ if and only if $p\equiv \pm r^2 \,({\rm mod}\,4q)$, where $r$ is an odd integer prime to $q$; and

\noindent 3) $\eta(2)=(-1)^{\frac{p^2-1}{8}}.$
\end{lem}

The following lemma then follows from Lemma \ref{11} and the content on the page 55 of \cite{Ireland}.

\begin{lem}\label{12} Let $p$ be an odd prime. Then

\noindent 1) $\eta(2)=1$ in $\mathbb F_p$  if and only if $p\equiv\pm 1\,({\rm mod}\,8)$;

\noindent 2) $\eta(3)=1$ in $\mathbb F_p$ if and only if $p\equiv\pm 1\,({\rm mod}\,12)$;

\noindent 3) $\eta(5)=1$ in $\mathbb F_p$ if and only if $p\equiv \pm1\,({\rm mod}\,5)$;

\noindent 4) $\eta(6)=1$ in $\mathbb F_p$ if and only if $p\equiv \pm1,\pm5\,({\rm mod}\,24)$; and

\noindent 5) $\eta(7)=1$ in $\mathbb F_p$ if and only if $p\equiv \pm1,\pm9,\pm25\,({\rm mod}\,28)$.
\end{lem}
\begin{proof} Cases 1), 2), 3) and 4) can be found on the page 55 of \cite{Ireland}. We now prove case 5). By Lemma \ref{11}, $\eta(7)=1$ in $\mathbb F_p$ if and only if $p\equiv \pm r^2 \,({\rm mod}\,28)$, where $r$ is an odd integer prime to 7. Thus $r=1,3,5,9,11,13,15,17,19,23,25,27$. And thus $r^2\equiv 1,9,25\,({\rm mod}\,28)$.  This completes the proof.
\end{proof}

\begin{theorem}\label{thm5} Let $p\equiv 1\,({\rm mod}\,4) $ be an odd prime and $m$ be an odd  positive integer. Let $u=\frac{p^m+1}{2}$ and $v$ be defined by \eqref{v=-2}. Then   $C_p(0,1,u^{-1}v)$ has parameters $[p^m-1,p^m-2m-2,4]$ if and only if $x^2+2(b+1)x+b=0$ has no solution $x\in\mathbb F_{p}^*\backslash \{\pm1\}$ such that $\eta(x)=1$ for any $b\in \mathbb F_{p}^*\backslash\{-1\}$.
\end{theorem}
\begin{proof} By the theory of multiplies of cyclic codes, $C_p(0,1,u^{-1}v)$ is equivalent to $C_p(0,u,v)$. Hence we only need to prove that $C_p(0,u,v)$ has parameters $[p^m-1,p^m-2m-2,4]$ if and only if $x^2+2(b+1)x+b=0$ has no solution $x\in\mathbb F_{p}^*\backslash \{\pm1\}$ such that $\eta(x)=1$ for any $b\in \mathbb F_{p}^*\backslash\{-1\}$.

Since $p\equiv1\,({\rm mod}\,4)$ and $m$ is odd, $u$ is odd. This together with even $v$ leads to $v\notin C_u$. Since  ${\rm gcd}(v,p^m-1)={\rm gcd}(-2,p^m-1)=2$,  $\mid C_v\mid=m$ due to Lemma \ref{1}. By Lemma \ref{lem2}, the dimension of $C_p(0,u,v)$ is equal to $p^m-2m-2$.  Note that $p\equiv1\,({\rm mod}\,4)$. According to Lemma \ref{weight2},  $C_p(0,u,v)$ has no codeword of Hamming weight 2.  $C_p(0,u,v)$ has no codeword of Hamming weight 3 if and only if  \eqref{3'} has no distinct solutions $x,y \in \mathbb F_{p^m}^*\backslash \{1\}$ for any $b_1,b_2\in \mathbb F_p^*\backslash \{-1\}$.
The proof is divided into the following two cases due to the symmetry.

\textit{Case 1.}  $(\eta(x),\eta(y))=(1,1)$: In this case, the second equation  subtracting the first equation  of  \eqref{3'} gives that $b_2=\frac{1-x}{y-1}$. The third equation subtracting the first equation  of  \eqref{3'} gives that $x^{-2}-1+b_2(y^{-2}-1)=0$.
Replacing $b_2$ by $\frac{1-x}{y-1}$ in $x^{-2}-1+b_2(y^{-2}-1)=0$ yields that
\begin{equation}\label{22}
(x+1)y^2-x^2(y+1)=0.
\end{equation}
If $x+1=0$, then $y+1=0$. And then $x=y$. This is contrary to that $x\neq y$. It then follows from \eqref{22} that
$(y-x)(y+\frac{x}{x+1})=0$. Hence $y=-\frac{x}{x+1}$. Replacing $y$ by $-\frac{x}{x+1}$ in $b_2=\frac{1-x}{y-1}$ gives that $x^2-2b_2x-b_2-1=0$. It can be reduced to $x^2+2(b_1+1)x+b_1=0$ thanks to $b_2=-b_1-1$. Since $p^m\equiv 1\,({\rm mod}\,4)$, $\eta(-1)=\eta(\alpha^{\frac{p^m-1}{2}})=1$. Since $y=-\frac{x}{x+1}$ and $\eta(-1)=\eta(x)=\eta(y)=1$,  $\eta(x+1)=1$. Therefore we have that $x^2+2(b_1+1)x+b_1=0$, where $x\neq \pm 1$, $\eta(x)=\eta(x+1)=1$ and $b_1\in \mathbb F_{p}^*\backslash\{-1\}$.

\textit{Case 2.}  $(\eta(x),\eta(y))=(1,-1)$: Similar as in case 1, in this case \eqref{3'} is reduced to  $(y+x)(y-\frac{x}{x+1})=0$. If $y=-x$, then $\eta(y)=\eta(-x)=\eta(x)=1$. This is contrary to that $\eta(y)=-1$. Hence $y=\frac{x}{x+1}$. Putting $y=\frac{x}{x+1}$ and the first equation  into the second equation of \eqref{3'} gives that $x^2+2(b_1+1)x+b_1=0,$ where $x\neq \pm 1$, $\eta(x)=1$ and $\eta(x+1)=\eta(\frac{x}{x+1})=\eta(y)=-1$.

It then follows from case 1 and case 2 that $C_p(0,u,v)$ has no codeword of Hamming weight 3 if and only if
\begin{equation}\label{23}
x^2+2(b_1+1)x+b_1=0,
\end{equation}
 has no solution $x\in\mathbb F_{p^m}^*\backslash \{\pm1\}$ such that $\eta(x)=1$ for any $b\in \mathbb F_{p}^*\backslash\{-1\}$.
Since $m$ is odd, \eqref{23} has no solution $x\in\mathbb F_{p^m}^*\backslash \{\pm1\}$ if and only if \eqref{23} has no solution $x\in\mathbb F_{p}^*\backslash \{\pm1\}$.
Then the desired result  follows immediately. This completes the proof.
\end{proof}

\begin{cor}\label{cor7} Let $p=5$ and $m$ be odd. Suppose $u=\frac{p^m+1}{2}$ and $v$ be defined by \eqref{v=-2}. Then  $C_p(0,1,u^{-1}v)$  is optimal and has parameters $[p^m-1,p^m-2m-2,4]$.
\end{cor}

\begin{proof} By Theorem \ref{thm5}, we only need to prove that $x^2+2(b+1)x+b=0$ has no solution $x\in\mathbb F_{p}^*\setminus\{\pm 1\}$ such that $\eta(x)=1$ as  $b$ runs through $\mathbb F_{p}^*\backslash\{-1\}$. It is straightforward to check that $x^2+2(b+1)x+b=0$ is irreducible over $\mathbb F_{p}^*\backslash\{-1\}$ for any $b\in \mathbb F_{p}^*\backslash\{-1\}$. According to Lemma \ref{irreducible}, $x^2+2(b+1)x+b=0$ remains irreducible over $\mathbb F_{p^m}^*$ since $m$ is odd. Hence $x^2+2(b+1)x+b=0$ has no solution $x\in\mathbb F_{p}^*\setminus\{\pm 1\}$ for any $b\in \mathbb F_{p}^*\backslash\{-1\}$. This completes the proof.
\end{proof}

\begin{remark}
Let $u=\frac{p^m+1}{2}$ and  $v$ be defined by \eqref{v=-2}, where $p\geq5$ is an odd prime and $m$ is an odd positive integer. According to Theorem \ref{thm5} and Lemma \ref{12}, the code $C_p(0,1,u^{-1}v)$ is not optimal and has parameters $[p^m-1,p^m-2m-2,3]$ if
$p\equiv 13\,({\rm mod}\,24)$ or $p\equiv 29,37,53\,({\rm mod}\,56)$ or
$p\equiv 29,53,65\,({\rm mod}\,84)$ or $p\equiv  37,193,253,277,337,373,$ $ 173,257,269,293, 353,377 ({\rm mod}\,420)$.
\end{remark}

An example of the code of Theorem \ref{thm5} is the following.

\begin{example} Let $p=5$ and $m=3$. Then $u=\frac{5^m+1}{2}=63$ and $v=122$. And then $u^{-1}v=122$. Let $\alpha$ be the generator of $\mathbb F_{5^m}^*$ with $\alpha^3+3\alpha+3=0$. The code $C_5(0,1,u^{-1}v)$ has parameters $[124,117,4]$  and generator polynomial $x^7 + 3x^6 + x^4 + x^3 + 4x^2 + 2x + 3$.
\end{example}

\section{Conclusions}

Let $p\geq 5$ be an odd prime and $u=\frac{p^m+1}{2}$. In coding theory, a fundamental problem is to construct optimal  cyclic codes achieving  a type of bound of linear codes.
Researchers have always been interested in
constructing optimal cyclic codes with respect to the Sphere Packing bound.
By the Sphere Packing bound, $p$-ary cyclic codes with parameters $[p^m-1,p^m-2m-2,4]$ are optimal.
Up to the authors' knowledge, there are only six classes of known optimal $p$-ary cyclic codes with parameters $[p^m-1,p^m-2m-2,4]$ for any odd prime $p\geq 5$.
In order to enrich the coding theory, this paper focuses on constructing optimal $p$-ary cyclic codes with parameters $[p^m-1,p^m-2m-2,4]$. Four classes of optimal $p$-ary cyclic codes $C_p(0,1,u^{-1}v)$ or $C_p(0,1,uv^{-1})$ with parameters $[p^m-1,p^m-2m-2,4]$ were presented by  weakening the necessary and sufficient conditions on  cyclic codes to have codewords of Hamming weight 3 and analyzing the solutions of certain equations over finite fields.
Moreover, three  classes of optimal $p$-ary cyclic codes constructed in this paper are infinite. For the convenience of  readers,
known optimal $p$-ary cyclic codes $C_p(0,1,v)$ with parameters $[p^m-1, p^m-2m-2,4]$   are listed in Table 1 and our optimal $p$-ary cyclic codes $C_p(0,1,u^{-1}v)$ or $C_p(0,1,uv^{-1})$ with parameters $[p^m-1, p^m-2m-2,4]$ are listed in Table 2.


\begin{table}
\centering
\caption{Our cyclic codes $C_p(0,1,u^{-1}v)$ or $C_p(0,1,uv^{-1})$  with parameters $[p^m-1, p^m-2m-2,4]$}
\label{Tab:02}
\begin{tabular}{c c c}
  \hline
\makecell[l]{Values of $v$ or requirements on $v$}  &  \makecell[l]{Conditions} & Reference \\[0.5mm]\hline
\makecell[l]{ $v(p^k+1)\equiv p^h+1\,({\rm mod}\,p^m-1)$} & \makecell[l]{${\rm gcd}(h-k,m)={\rm gcd}(h+k,m)=1$,\\ $p^m\equiv 1 \,({\rm mod}\,4)$ or ${\rm gcd}( v,p^m-1)=1$,\\
\eqref{weight3} has no distinct solutions $x,y$ in $\mathbb F_{p}^*\backslash \{1\}$} & Theorem \ref{thm2} \\[0.4cm]
\makecell[l]{ $v(p^k-1)\equiv p^h-1\,({\rm mod}\,p^m-1)$} &  \makecell[l]{${\rm gcd}(h,m)={\rm gcd}(h-k,m)=1$,\\ $p^m\equiv 1 \,({\rm mod}\,4)$ or ${\rm gcd}( v,p^m-1)=1$,\\
\eqref{weight3} has no distinct solutions $x,y$ in $\mathbb F_{p}^*\backslash \{1\}$} & Theorem \ref{thm3} \\[0.4cm]
\makecell[l]{ $v=\frac{p^m-1}{2}+p^k+1$} &  \makecell[l]{$\frac{m}{{\rm gcd}(k,m)}\equiv 0\,({\rm mod}\,4)$} & Theorem \ref{thm4} \\[0.5mm]
\makecell[l]{ $v=\frac{p^m-1}{2}+2p^k$} &  \makecell[l]{$m\equiv 0\,({\rm mod}\,4)$}&Theorem \ref{thm4} \\[1mm]
\makecell[l]{$v=\frac{p^m-1}{2}+\frac{p^m-3}{2}$} & \makecell[l]{ $m\equiv 0\,({\rm mod}\,2)$}&Theorem \ref{thm4}  \\[0.5mm]
\makecell[l]{$v=\frac{p^m-1}{2}-1$}& \makecell[l]{$p^m\equiv 1 \,({\rm mod}\,4)$ or ${\rm gcd}( v,p^m-1)=1$}&Theorem \ref{thm4}  \\[0.4cm]
\makecell[l]{$v\equiv -2\,({\rm mod}\,p^m-1)$}&\makecell[l]{$p\equiv 1\,({\rm mod}\,4)$, $m$ is odd, $b\in \mathbb F_{p}^*\backslash \{-1\}$,\\ $x^2+2(b+1)x+b=0$ has no solution\\ $x\in\mathbb F_{p}^*\backslash \{\pm1\}$ such that $\eta(x)=1$}&Theorem \ref{thm5}  \\[0.4cm]
\makecell[l]{$v(5^k+1)\equiv 5^h+1\,({\rm mod}\,5^m-1)$} & \makecell[l]{$m$ is odd, $v\equiv 1\,({\rm mod}\,4)$,\\${\rm gcd}(h-k,m)={\rm gcd}(h+k,m)=1$} & Corollary \ref{cor2} \\[0.5mm]
\makecell[l]{$v(5^k+1)\equiv 5^h+1\,({\rm mod}\,5^m-1)$} & \makecell[l]{$m$ is even, $v\equiv 3\,({\rm mod}\,4)$,\\ ${\rm gcd}(h-k,m)={\rm gcd}(h+k,m)=1$} & Corollary \ref{cor2} \\[0.5mm]
\makecell[l]{$v(5^k-1)\equiv 5^h-1\,({\rm mod}\,5^m-1)$} &  \makecell[l]{$m$ is odd, $v\equiv 1,2\,({\rm mod}\,4)$,\\ ${\rm gcd}(h,m)={\rm gcd}(h-k,m)=1$} & Corollary \ref{cor4} \\[0.5mm]
\makecell[l]{$v(7^k-1)\equiv 7^h-1\,({\rm mod}\,7^m-1)$} &  \makecell[l]{$m$ is odd, $v\equiv 5\,({\rm mod}\,6)$, ${\rm gcd}( v,p^m-1)=1$\\ ${\rm gcd}(h,m)={\rm gcd}(h-k,m)=1$} & Corollary \ref{cor5} \\[0.5mm]
\makecell[l]{$v\equiv -2\,({\rm mod}\,5^m-1)$}&\makecell[l]{$m$ is odd}&Corollary \ref{cor7}  \\[0.5mm]
\hline
\end{tabular}
\end{table}

\section{Acknowledgements}{Jinmei Fan was supported by National Natural Science Foundation of China (Grant No. 12061027), Yanhai Zhang was supported by the
Doctoral Research Foundation of Guilin University of Technology (Grant No. GUTQDJJ2018033) and by the Opening Fund of Key Laboratory of Cognitive Radio and Information Processing, Ministry of Education (Grant No. CRKL210206).}





\begin{thebibliography}{99}

\bibitem{MacWilliams} MacWilliams F, Sloane N. The theory of error-correcting codes. Amsterdam-New York-Oxford: North-Holland Publishing Company, 1977

\bibitem{Chien} Chien R.
Cyclic decoding procedure for the Bose-Chaudhuri-Hocquenghem codes.
IEEE Trans Inf Theory, 1964, 10(4): 357--363



\bibitem{Forney}
 Forney G D. On decoding BCH codes.
 IEEE Trans Inf Theory, 1965, 11(4): 549--557

\bibitem{Prange}
 Prange E. Some Cyclic Error-Correcting Codes with Simple Decoding Algorithms,
 AFCRC-TN-58-156. Cambridge: Mass, 1985

\bibitem{DH}
 Ding C S, Helleseth T.
Optimal ternary cyclic codes from monomials.
 IEEE Trans Inf Theory, 2013, 59(9): 5898--5904



\bibitem{CDY}
 Carlet C, Ding C S, Yuan J, et al.
Linear codes from highly nonlinear functions and their secret sharing schemes.
 IEEE Trans Inf Theory, 2005, 51(6): 2089--2102


\bibitem{LLHDT} Li N, Li C L, Helleseth T, et al. Optimal ternary cyclic codes with minimun distance
four and five. Finite Fields Appl, 2014, 30: 100--120

\bibitem{LLHD}
 Li C L, Li N, Helleseth T,  et al.
The weight distributions of several classes of cyclic codes from APN monomials.
 IEEE Trans Inf Theory, 2014, 60(8): 4710--4721


\bibitem{open1}Li N, Zhou Z C, Helleseth T. On a conjecture about a class of optimal ternary cyclic codes. In 2015 Seventh International Workshop on Signal Design and its Applications in Communications (IWSDA), 2015,  10: 62--65


 \bibitem{FLZ}
 Fan C L, Li N, Zhou Z C.
A class of optimal ternary cyclic codes and their duals.
Finite Fields Appl, 2016, 37: 193--202

\bibitem{WW}
Wang L S, Wu G F.
Several classes of optimal ternary cyclic codes with minimal distance four. Finite Fields Appl, 2016, 40: 126--137

\bibitem{YZD}
 Yan H D, Zhou Z C,  Du X N.
A family of optimal ternary cyclic codes from the Niho-type exponent.
Finite Fields Appl, 2018, 54: 101--112

\bibitem{open2} Han D C, Yan H D. On an open problem about a class of optimal ternary cyclic codes. Finite Fields
Appl, 2019, 59: 335--343

\bibitem{Zha1}
Zha Z B,  Hu L.
New classes of optimal ternary cyclic codes with minimum distance four. Finite Fields Appl, 2020, 64: 101671

\bibitem{open3} Liu Y, Cao X W, Lu W. On some conjectures about optimal ternary cyclic codes. Des Codes Cryptogr, 2020, 88(2): 297--309


\bibitem{Zha2}
 Zha Z B, Hu L, Liu Y, et al.
Further results on optimal ternary cyclic codes.
Finite Fields Appl, 2021, 75: 101898

\bibitem{LCL}
Liu Y, Cao X W, Lu W.
 Two classes of new optimal ternary cyclic codes.
Adv  Math Commun, 2023,
17(4): 979--993

\bibitem{ZLS}
 Zhao H, Luo R, Sun T J.
 Two families of optimal ternary cyclic codes with minimal distance four.
 Finite Fields Appl, 2022, 79: 101995

\bibitem{XCX}
 Xu G K, Cao X W, Xu S D.
Optimal $p$-ary cyclic codes with minimum distance four from monomials.
 Cryptogr Commun,  2016, 8: 541--554



\bibitem{Fan}
 Fan J M, Zhang Y H.
Optimal quinary cyclic codes with minimum distance four.
Chinese J Electron, 2020, 29(3): 515--524

\bibitem{Liu1}
Liu Y,  Cao X W.
Four classes of optimal quinary cyclic codes.
 IEEE Commun Lett, 2020, 24(7): 1387--1390


\bibitem{WLZ}
Wu G F, Liu H, Zhang Y Q.
Several classes of optimal $p$-ary cyclic codes with minimal distance four.
Finite Fields Appl, 2023, 92: 102275

\bibitem{YCD} Yuan J, Carlet C,  Ding C S, et al. The weight distribution of a class
of linear codes from perfect nonlinear functions. IEEE Trans Inf Theory, 2006, 52(2): 712--717

\bibitem{Gong}
 Golomb S W, Gong G.
Signal design for good correlation: for wireless communication, cryptography, and radar.  New York:
 Cambridge University Press, 2005


\bibitem{FXXZ}	Fan J M, Xu Y G, Xia Y B, et al. Two families of Niho sequences having four-valued cross correlation with $m$-sequences. Sci China Math, 2017, 60(12): 2377--2390

\bibitem{Lidl}
 Lidl R, Niederreiter H.
 Finite Fields. Cambridge:
 Cambridge University Press, 1997

\bibitem{Ireland}  Ireland K, Rosen M. A Classical Introduction to Modern Number Theory. New York: Springer, 1990




























\end{thebibliography}
\end{document}